\documentclass[conference]{IEEEtran}
\usepackage{amssymb}
\usepackage{verbatim}
\usepackage{graphicx,epsfig}
\usepackage{amsfonts}
\usepackage{subfigure}
\usepackage[normalem]{ulem}
\usepackage{amsmath}
\usepackage{amsthm}

\def\ps@headings{%
\def\@oddhead{\mbox{}\scriptsize\rightmark \hfil \thepage}%
\def\@evenhead{\scriptsize\thepage \hfil \leftmark\mbox{}}%
\def\@oddfoot{}%
\def\@evenfoot{}}
\makeatother 

\def\bi{\begin{itemize}}
\def\ei{\end{itemize}}
\def\bequ{\begin{equation}}
\def\eequ{\end{equation}}

\def\benum{\begin{enumerate}}
\def\eenum{\end{enumerate}}

\begin{document}

\title{Primer and Recent Developments on Fountain Codes}


\author{\IEEEauthorblockN{Jalaluddin Qureshi\IEEEauthorrefmark{1}, Chuan Heng Foh\IEEEauthorrefmark{2}, Jianfei Cai\IEEEauthorrefmark{1}}
\IEEEauthorblockA{\IEEEauthorrefmark{1}
School of Computer Engineering\\
Nanyang Technological University, Singapore}%
\IEEEauthorblockA{\IEEEauthorrefmark{2}
Centre for Communication Systems Research \\
Department of Electronic Engineering\\
University of Surrey\\
Surrey GU2 7XH, United Kingdom} }

\maketitle

\begin{abstract}
In this paper we survey the various erasure codes which had been proposed and patented recently, and along the survey we provide introductory tutorial on many of the essential concepts and readings in erasure and Fountain codes. Packet erasure is a fundamental characteristic inherent in data storage and data transmission system. Traditionally replication/ retransmission based techniques had been employed to deal with packet erasures in such systems. While the Reed-Solomon (RS) erasure codes had been known for quite some time to improve system reliability and reduce data redundancy, the high decoding computation cost of RS codes has offset wider implementation of RS codes. However recent exponential growth in data traffic and demand for larger data storage capacity has simmered interest in erasure codes. Recent results have shown promising results to address the decoding computation complexity and redundancy tradeoff inherent in erasure codes.
\end{abstract}

\let\thefootnote\relax\footnote{This paper appears in BSP Recent Patents on Telecommunications.}

\newtheorem{theorem}{\textbf{Theorem}}
\newtheorem{axiom}{\textbf{Axiom}}
\newtheorem{lemma}{\textbf{Lemma}}

\section{Introduction}\label{sec:introduction}

The aim of this paper is to provide a self-contained introduction about the erasure coding scheme and recent developments in Fountain codes, which are currently the predominant class of erasure codes.

Packet erasure is one the fundamental and inevitable characteristic in data transmission and data storage system. For example routers may drop a packet due to congestion. Similarly a file in a data storage system can be erased due to component failures. The problem of packet erasure exuberates for data transmission on wireless channel due to the shared medium of transmission resulting in packet collisions. In addition to packet collision, for wireless channel, packet may also be erased due to channel fading, additive white Gaussian noise (AWGN) and signal attenuation. The average wireless channel erasure rate in some deployments can be as high as 20-50\%~\cite{Aguayo04,Rozner07}.

Traditional approach of dealing with packet erasure is to use replication and retransmission. The method of replication and retransmission introduces control overhead. For data storage system, replication provides limited reliability. For instance in the event that the original file and the replicated files are both erased, then the data storage system can not recover the original file. Similarly the use of retransmission technique for data transmission system is dependent on packet acknowledgement control frame from the client. It is also possible that an acknowledgement frame can also be erased due to the same reasons as the original data packet, erroneously resulting in the retransmission of those data packets which the client has already received.

For wireless networks, the transmission of acknowledgement frame occupies the wireless channel medium and therefore adversely affects the transmission bandwidth. The problem of collecting acknowledgement exacerbates when the transmitter is multicasting the data to $N$ clients, in which case it has to collect $N$ acknowledgement frames. Due to the exacerbation to efficiently collect ACK frames for multicast transmission, and efficiently retransmit the erased packets, legacy IEEE 802.11 multicast transmission is a ``no-ACK, no-retransmission'' scheme, in which the access point (AP) transmits the data packet and then waits for the channel to be free, conforming only to the carrier sense multiple access collision avoidance - (CSMA/CA) access procedure, before transmitting the next data packet.

The urgency to have a reliable multicast transmission scheme for 802.11 wireless networks is reflected by the recent formation of the IEEE working group, called Task Group aa (TGaa)~\cite{Maraslis12}. The wireless multicast technology is also employed by Multimedia Broadcast Multicast Service (MBMS) for 3GPP cellular network (3GPP TS 26.346).

Erasure codes have been proposed as an efficient remedy to improve the reliability and scalability of data transmission over erasure channels. In an erasure codes transmission scheme the client can recover the $k$ data packets from the $n$ transmitted packets, where $n=(1+\epsilon)k$, at a \textit{rate} $r$ given as $n/k$, and $\epsilon$ is known as the \textit{overhead} (or \textit{redundancy}) of the coding scheme. For recovering the input packets, it is irrelevant which packets the client had received, as long as it has received any $k$ linearly independent packets, the client can decode the $k$ input packets. In erasure coding, the coded packets are generated by linearly mapping the packets with coefficients from a finite field $\mathbb{F}_q$, where for computer science applications the \textit{finite field size} $q$ is given as $q=2^{i}$, $i\in \mathbb{N}_1$, and $\mathbb{N}_1$ is the set of natural numbers excluding zero. Erasure codes, such as Triangular codes, may also be non-linear, and therefore may involve operations over real field.

Erasure coding is a technique to provide reliability in an event of packet erasure. Error coding and erasure coding are class of \textit{Forward Error Correction} (FEC) techniques. \textit{Error coding} scheme protects the system from data corruption, e.g. noise or attenuation, whereas \textit{erasure coding} protect the system from data lost during transmission, e.g. packet drop at router due to congestion or packet collision over the wireless channel. In our paper we only consider a \textit{binary erasure channel} (BEC) model. In a BEC a transmitted packet is either correctly received with probability $p$, or not received with probability $1-p$. Similarly a stored data packet is either not corrupted with probability $p$, or corrupted with probability $1-p$. The BEC is also known as Bernoulli channel model. When considering data transmission to multiple clients, i.e. multicast or broadcast, the packet erasure probability on each of the transmission channel is assumed to \textit{independent and identically distributed} (iid).

To illustrate how erasure coding can be improve the reliability of a transmission system, consider for example a wireless network where an AP is multicasting packets $c_1$ and $c_2$ to clients $R_1$ and $R_2$. However, $R_1$ receives $c_1$ but not $c_2$, whereas $R_2$ receives $c_2$ but not $c_1$. In this case, rather than retransmitting packet $c_1$ and $c_2$ in two different time slots, it is possible for the transmitter to encode the packets $c_1\oplus c_2$ over $\mathbb{F}_2$, and transmit the encoded packet in one time slot. On receiving the encoded packet both the client can recover the lost packet by decoding the original packet with the encoded packet. This therefore reduces the number of retransmissions from two rounds to one, and hence improves the network bandwidth.

Similarly erasure coding can also improve the reliability performance in data storage system. Consider for example packet $c_1$ and $c_2$ being stored on a data storage system. To improve the reliability of the system, traditional approach would replicate $c_1$ and $c_2$, and store two copies of $c_1$ and two copies of $c_2$. However in an event of storage failure, where two copies are erased, and both these erased copies happen to be $c_1$. In this case there is no way the system can recover packet $c_1$. On the contrary, if instead of storing the replicated packet, the system stores two coded packets $c_1+c_2$ and $c_1+\alpha c_2$, then in an event of storage failure, where two packets are erased, then irrespective of which these two erased packets are, the system can still recover $c_1$ and $c_2$ from the remaining two packets.

Such reliability gains however do not come without tradeoffs. While linear coding schemes over larger field size such as the random linear (RL) codes and Reed-Solomon (RS) codes can deliver optimal rate\footnote{Strictly speaking, RL codes are suboptimal, however the difference between optimal rate and the rate of RL codes over large field size $\mathbb{F}_{q\geq256}$ is negligible, and in order of $\approx$10$^{-4}$.}, decoding RS and RL coded packets are costly, requiring the use of matrix inversion which is implemented using Gaussian elimination with complexity $\mathcal {O}(k^3)$ when the coding coefficients are dense, or variants of the Wiedemann algorithm with complexity $\mathcal {O}(k^2\log k)$ when the coding coefficients are sparse~\cite{Vasilenko06}, in addition to the cost of matrix multiplication.

It has been further shown that the decoding computation cost is also dependent on the field size from which the coding coefficients are selected. Practical implementation of RL codes on iPhone 3G has shown that the decoding throughput of RL codes over $\mathbb{F}_2$ is approximately six times faster than decoding RL codes over $\mathbb{F}_{256}$ on the same testbed. Similarly encoding over $\mathbb{F}_2$ is approximately eight times faster compared to encoding over $\mathbb{F}_{256}$~\cite{Vingelmann10}. Unfortunately smaller field size can not be used for RL codes, as larger field size is a prerequisite for RL codes to deliver optimal rate.

Experimental evaluation of RL codes over $\mathbb{F}_{256}$ on iPhone 3G, for $k=64$ with packet length of 4096 bytes, has shown that for two devices with same configurations and running the same applications, the device running with an additional RL decoding application consumes approximately $20\%$ more battery energy reserves~\cite{Shojania09}. Mobile phone batteries suffer from severe energy limitation, which is why handset vendors are increasingly interested in energy optimization of various smartphone applications which can sustain longer operational time.

The high decoding cost of packets coded over large field size can be addressed by using the simpler XOR addition for encoding and decoding, which is also known as $\mathbb{F}_2$ addition. It has been shown that XOR addition of two packets, each 1000 bytes long only consumes 191 nJ of energy~\cite{Vingelmann09}. Given that the transmission of a packet of the same length over IEEE 802.11 network on Nokia N95 consumes 2.31 mJ of energy~\cite{Vingelmann09}, the overall energy cost of XOR addition has no apparent affect on the operational time of a mobile phone.

It is apparent from the above discussion, that an amicable solution is needed to address the throughput performance and decoding computation cost tradeoff inherent in erasure coding. In the last decade a subclass of erasure codes, known as Fountain codes have gained widespread acceptance for its ability to address this tradeoff. However Fountain codes assume very large input packet length to deliver near-optimal transmission rate using linear decoding algorithm. For example, even for very large input packet length of $k\approx 10,000$, LT codes, which is an implementation of Fountain codes,  have an overhead of about 5\%~\cite{Mackay05}. Similarly for $k=65,536$, Raptor codes, which is also an implementation of Fountain codes, have an overhead of about 3.8\%~\cite[Table 1]{Amin06}.

More recently, Qureshi \textit{et al.} proposed the Triangular coding scheme~\cite{Qureshi12} to address the tradeoffs in performance, computation costs and packet length in erasure codes. An ideal erasure code should be able to deliver near-optimal transmission rate with linear computation cost, with such performance being independent of any parameter including input packet length. As we will show, the Triangular codes have the potential to be designed to fulfill the characteristics of such ideal codes.

The rest of the paper is organized as follow. We first present a tutorial on erasure codes in Section~\ref{sec:background}, and the characteristics of classical erasure coding schemes, the Reed-Solomon codes, low density parity check codes, and random linear codes in Section~\ref{sec:classical}. We then present the highlight and performance of Fountain codes, namely Tornado codes, Luby-Transform codes, Raptor codes and standardized Raptor codes in Section~\ref{sec:fountain}, and those of Triangular codes in Section~\ref{sec:triangular}. We then conclude with open research problems in erasure coding in Section~\ref{sec:future} and summary of our work in Section~\ref{sec:conclusion}.

\section{Technical Background}\label{sec:background}
We provide a brief mathematical background on the encoding and decoding for linear codes that covers Fountain codes. Let there be $k$ numbers of input data packets to be encoded for transmission. The set of $k$ input data packets is given by the vector $\textbf{M}=[\textbf{c}_1, \textbf{c}_2, ..., \textbf{c}_k]$, and the set of encoded packets transmitted by the server is given by the vector $\textbf{X}=[\textbf{x}_1, \textbf{x}_2, ..., \textbf{x}_j]$. If $\textbf{M}\subseteq\textbf{X}$, then such a coding scheme is known as \textit{systematic codes}, and \textit{non-systematic codes} otherwise. The set of innovative packets which the client has received is given as $\textbf{Y}_{i}$, $\textbf{Y}_{i}\subseteq \textbf{X}$. We define a packet as innovative packet, if the packet is linearly independent with respect to all the packets that a client already has.

\subsection{Linear Scalar Codes} \label{subsect:linear}
A \textit{linear} code $\mathcal{C}$ of length $B$ over the finite field $\mathbb{F}_q$ is formally defined as a linear subspace of the vector space $\mathbb{F}_q^B$, and \textit{non-linear} otherwise. The code $\mathcal{C}$, $\mathcal{C}\subseteq \mathbb{F}_q^B$, is a linear subspace of $\mathbb{F}_q^B$, if for all $\textbf{x}_i, \textbf{x}_j\in \mathcal{C}$, $\textbf{x}_i\oplus \textbf{x}_j\in \mathcal{C}$ and for all $\textbf{g}_m\in \mathbb{F}_q$, $\textbf{x}_i\in \mathcal{C}$, $\textbf{g}_m\cdot \textbf{x}_i\in \mathcal{C}$. Therefore the subset code $\mathcal{D}=\{111, 010\}$ of $\mathbb{F}_2^3$ is not a linear code, as $111\oplus 010=101\notin \mathcal{D}$.

In this paper, we only consider scalar erasure codes. In \textit{scalar coding}, a packet can not be split into smaller sub-packets, whereas in \textit{vector coding}, a packet may be split into smaller sub-packets. Vector codes are shown to outperform scalar codes in certain classes of transmission problem using erasure codes such as the cooperative data exchange problem~\cite{Tajbakhsh11}.

\subsection{Packet Encoding} \label{subsect:encoding}
To encode a packet, the server chooses coefficients vectors from a finite field $\mathbb{F}_q$, to form an encoding coefficient vector $\textbf{G}_j=[\textbf{g}_1, \textbf{g}_2, ...,\textbf{g}_k]$, $\textbf{g}_k\in \mathbb{F}_q$, which is then multiplied with vector $\textbf{M}$ to generate a single coded packet $\textbf{x}_j$ given as,

\begin{equation}\label{eq:coding}
\textbf{x}_j = \textbf{G}_j \textbf{M}^T.
\end{equation}

Coding over field size $\mathbb{F}_{q>2}$, requires finite field multiplication and XOR addition operations. Whereas for XOR coding, $\textbf{g}_k\in \{0,1\}$, only the addition operation is required. For the client to be able to decode the coded packets, the transmitter needs to add information about the vector $\textbf{G}_j$ in the packet header. The number of bits required to store one such coefficient $\textbf{g}_k$ is given as $q$ bits. Since there are $k$ such coefficients, the packet overhead for a coded packet is given as $k\times q$ bits.

\subsection{Packet Decoding} \label{subsect:decoding}

After a client $R_i$ has received $k$ innovative packets, these $k$ packets are placed in matrix $\textbf{Y}_i$. The coding coefficients from all the coded packet's header is used to form a $k\times k$ coefficient matrix $\textbf{H}_i$. The set of original packets $\textbf{k}$ can then be decoded by $R_i$ as,

\begin{equation}\label{eq:decoding}
\textbf{M}^T = \textbf{H}_i^{-1} \textbf{Y}_i^{T} .
\end{equation}

Inversion of $\textbf{H}_i$ can be performed using Gaussian elimination. We illustrate the above decoding process with the aid of simple example. Consider the three received coded packets given as, $x_1=c_1\oplus c_2$, $x_2=c_2\oplus c_3$ and $x_3=c_1\oplus c_2\oplus c_3$. The corresponding $\mathbb{F}_2$ matrix $\textbf{H}_i$, and its inverse, are given as,

\[
\textbf{H}_i = 
\begin{pmatrix}
1 & 1 & 0\\
0 & 1 & 1\\
1 & 1 & 1
\end{pmatrix},
\textbf{H}_i^{-1} = 
\begin{pmatrix}
0 & 1 & 1\\
1 & 1 & 1\\
1 & 0 & 1
\end{pmatrix}.
\]

\noindent One can verify that based on the inverted matrix, the decoded packets are given as, $c_1=x_2\oplus x_3$, $c_2=x_1\oplus x_2\oplus x_3$, and $c_3=x_1\oplus x_3$.

Coefficient matrix $\textbf{H}_i$ for XOR coded packets will form a $\mathbb{F}_2$ matrix, and only require the ``row interchange'' and ``row addition'' operations~\cite{Fraleigh87} for Gaussian elimination. After the inversion of the binary matrix, only the addition operation needs to be performed. However for packets coded over larger field size, Gaussian elimination process would also need to perform ``row scaling'' operation (i.e. multiplying a row of the matrix with a non-zero scalar) in addition to the ``row interchange'' and ``row addition'' operations. After the inversion, the client will need to perform multiplication and addition to decode the native packets. Therefore even though the complexity order of matrix inversion for both $\mathbb{F}_2$ matrix and $\mathbb{F}_{q>2}$ is same, decoding XOR coded packets requires fewer computation steps.

\subsection{Decoding Algorithms}\label{sec:decoding}
All linear erasure codes essentially use the Gaussian elimination or one of its variants to perform decoding. Gaussian elimination consists of two major steps, \textit{triangularization} of the matrix into an upper or lower Triangular matrix, with complexity $\mathcal {O}(k^3)$ for a $k\times k$ full rank matrix, and \textit{back-substitution} of the triangular matrix to solve the unknown variables with complexity $\mathcal {O}(k^2)$~\cite{Fraleigh87}. The belief-propagation (BP) decoding algorithm used for LT codes, and the inactivation decoding algorithm used for Raptor codes~\cite[pp. 247-255]{Raptor}, can therefore be entirely described using Gaussian
elimination steps of triangularization and back-substitution. We refer interested readers to~\cite[Chapter 11]{Vasilenko06} and the references therein for discussion on various methods to solve a system of linear equations.

When the Triangular matrix has an average sparsity of $\omega$, then the complexity of back-substitution is given as $\mathcal {O}(k\omega)$. It is interesting to note that the average sparsity $\omega$, of a Triangular matrix is bounded as $\frac{k+1}{2}$, which explains the derivation of total number of computation steps for back substitution as $\frac{k^2}{2}+\frac{k}{2}$ in~\cite{Fraleigh87}, i.e. $k$ times $\omega$. All efficient erasure coding scheme are essentially build on this concept of generating sparse coding coefficients, i.e. $\omega<\frac{k+1}{2}$, which can be solved using back-substitution. For LT and Raptor codes, $\omega$ is given by $\log(\frac{k}{\delta})$ and $\log(\frac{1}{\epsilon})$ respectively.

The work of Darmohray and Brooks~\cite{Darmohray87} published in 1987 has shown the performance merits of Gauss-Jordan elimination over Gaussian elimination in multi-processors system despite the former higher computation cost. Gauss-Jordan elimination method in a multi-processor system still continues to be of significant implementation importance in recent works~\cite{Park10}. While the lower computational complexity of Wiedemann algorithm to solve sparse system of linear equations has motivated the design of sparse coding vector~\cite{Sung11}, over $\mathbb{F}_{q>2}$.

\subsection{Tanner Graph}
Many of the characteristics of Fountain codes are illustrated using the Tanner graph, which is a special class of the bipartite graph. In a Tanner graph, the input packets and coded packets are arranged into two disjoint vertex sets, i.e. there can only exist an edge between vertices represented by an input packet and coded packet. The set of packets used to encode the coded packet is given by the edges connecting the coded packet, as each of these edges will be connected to an input packet. The Tanner graph can be treated as a pictorial representation of a $\mathbb{F}_2$ matrix $\textbf{H}_i$. An example of a Tanner graph is given in Figure~\ref{fig:tanner}.

A Tanner graph, and consecutively the codes that the Tanner graph represents, is said to be a regular graph (/regular code) if each of the coded packet is generated using a fixed number of input packets, and irregular graph (/irregular code) otherwise. 

\begin{figure}
\begin{center}
\includegraphics[width = 0.5\textwidth]{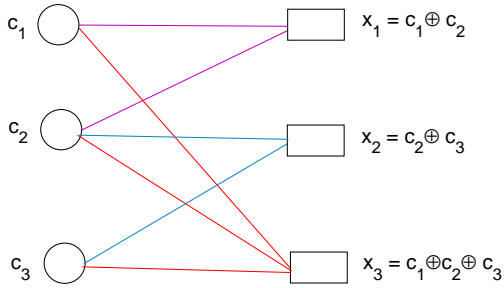}
\end{center}
\caption{An example of an irregular Tanner graph representing the irregular XOR codes given in Section~\ref{subsect:decoding}.} \label{fig:tanner}
\end{figure}

\subsection{Finite Field Arithmetic}
Packet encoding and decoding represented by Equation~\eqref{eq:coding}-\eqref{eq:decoding} are done using finite field arithmetic, which is different from the real field arithmetic. Fundamentally finite field arithmetic stipulates that for encoding and decoding over a field size $\mathbb{F}_q$, the operands are given as $\mathbb{F}_q=\{0,1,\ldots,q-1\}$, and the solution when these operands are added, subtracted, multiplied or divided, is also an element of the field $\mathbb{F}_q$. Because of this, multiplication in Equation~\eqref{eq:coding} does not increases the length of packet, which would have otherwise been observed for real field multiplication.

The monograph by Vasilenko, translated by Martsinkovsky~\cite{Vasilenko06}, covers many of the number-thoeretic algorithms over finite field in an easy-to-understand writing style.

\subsection{Related Problem}\label{subsect:related}
The erasure coding model assumes a single server broadcasting data to multiple clients over BEC, where all the clients want all the data being broadcast by the server. A coding scheme related to erasure code is the index code to solve the index coding problem~\cite{Chaudhry08}.

The \emph{index coding problem} is an instance of packet transmission problem to a set of $N$ clients $R = \{r_1, r_2, ..., r_N\}$ by the server having a set of $k$ packets $P = \{c_1, c_2, ..., c_k\}$, and the \emph{side information} about the set of packets each client $r_i\in R$ has, $H_i\subset P$, and the set of symbol each client wants, $W_i\subseteq P \backslash H_i$, such that the total number of symbols transmitted is minimized.

The erasure codes can therefore be treated as a special case of index codes, where $W_i=P$, $\forall i$, at the start of the transmission. Finding optimal index codes and approximation to such optimal index codes are both NP-hard for the general index coding problem.

Another related problem is that of \textit{network coding}, where multiple sources are multicasting data to multiple common clients connected to the sources through intermediate routers, with BEC model. In network coding, apart from the sources, intermediate routers can also perform encoding before forwarding the coded packets. For a self-contained introductory tutorial on network coding we refer readers to~\cite{Chou07}.

\section{Classical Erasure Codes}\label{sec:classical}
\subsection{Reed-Solomon Codes}
Some of the first class of erasure codes were based on the Reed-Solomon (RS) error coding scheme proposed by Reed and Solomon in 1960 as an error coding~\cite{Reed60} scheme, and presented as an erasure coding scheme by Rizzo in 1997~\cite{Rizzo97}. The RS codes have also been suggested for use in hybrid Automatic Repeat reQuest (ARQ) protocol~\cite[Chap 7]{Wicker94} to improve transmission reliability in wireless network.

RS code are linear codes, using coding coefficients $\textbf{G}_j$ from the Vandermonde matrix, $M_n$ where the rows are given by a geometric progression sequence. Any $k$ arbitrary rows from a $n\times k$ Vandermonde matrix, given as $M'_k$, form a nonsingular matrix, provided that the common ratio of each of the geometric progression sequence is unique. The client can decode the input packets once it has received any $k$ coded packets from $n$ transmissions, where $n\geq k$.

A Vandermonde matrix is given as follow,
\begin{align*}
M_n=
\begin{pmatrix}
1 & a_1 & a_1^2 & \ldots & a_1^{k-2} & a_1^{k-1}\\
1 & a_2 & a_2^2 & \ldots & a_2^{k-2} & a_2^{k-1}\\
  & \vdots &    & \ddots &           & \vdots\\
1 & a_{n} & a_{n}^2 & \ldots & a_{n}^{k-2} & a_{n}^{k-1}\\
\end{pmatrix},
\end{align*}

\begin{theorem}\label{theorem:1}
The matrix $M'_k$ is invertible if each of the number $a_\ell$ is unique.
\end{theorem}
\begin{proof}
A proof of this can be found in~\cite[pp. 17-18]{Mirsky55}.
\end{proof}

The main shortcoming of the RS codes is that decoding of RS codes requires the computationally expensive Gaussian elimination method, therefore limiting the size of the input packets. Since the RS codes use dense coding coefficients, computationally efficient algorithms to inverse the matrix such as structured Gaussian elimination and Wiedemann algorithms are of no help, as these algorithms run on sparse matrices.

Despite such shortcoming, RS codes still continues to be used for various applications such as redundant arrays of inexpensive disks (RAID) system~\cite{Plank97}, Microsoft Windows Azure cloud Storage~\cite{Calder11} and as part of the Microsoft Real-Time Streaming Protocol~\cite{RTSP12}, amongst others. 

\subsection{Low Density Parity Check Codes}
The low density parity check (LDPC) codes as the name implies are very sparse\footnote{The sparsity of a coding vector is denoted by $\omega$. A $\omega$-sparse coding vector is one with $k-\omega$ zero coefficients and $\omega$ non-zero coefficients.} random $\mathbb{F}_2$ linear codes, i.e. $\omega\ll k$, and the coding vector is selected randomly over $\mathbb{F}_2$. The LDPC codes were invented by Galleger in 1962. The importance of these codes has largely remained dormant until the 1990s when Fountain codes were designed using irregular LDPC codes. 

There are several decoding algorithms for LDPC, each for different channel model. For BEC, the message passing algorithm is used. The message passing algorithm is analogous to the back-substitution algorithm. A pseudocode of the message passing algorithm is given in~\cite[pp. 52]{Johnson10}.

\subsection{Random Linear Codes}
Random Linear (RL) codes, also known as Random Fountain codes and Random Linear Network Codes (RLNC), use coding coefficient $\textbf{g}_k$ randomly selected from $\mathbb{F}_q$. When the finite field from which the coding coefficients are randomly selected is large, then any $k$ received coded packets are linearly independent with respect to each other with high probability. Analytical results on the performance of RL codes for different finite field size is given in~\cite{Cruces11}.

The major drawback of RL codes is that it uses the computationally expensive Gaussian elimination method for decoding, which limits its implementation for larger input packet length~\cite{Shojania09, Vingelmann09}. On the positive aspect, because of the random encoding procedure, RL code is in particular popular in network coding~\cite{Yeung05}, as it can be implemented in a decentralized network system because of its random encoding process.

\section{Fountain Codes}\label{sec:fountain}

We first give a formal definition of Fountain codes, adopted from~\cite{Raptor}. An erasure code can be classified as a Fountain code, if it has the following traits.

\begin{itemize}
\item The number of coded packets which can be generated from a given set of input packets should be sufficiently large.
\item Irrespective of which packets the client has received, the client should be able to decode the $k$ input packets using any $k$ linearly independent received packets.
\item The encoding and decoding computation cost should be linear.
\end{itemize}

\noindent We now visit the development of various Fountain coding schemes.

\subsection{Tornado Codes}
Spurred by the popularity of the INTERNET in the 1990s and high decoding cost of RS codes, Luby \textit{et al.} developed the Tornado codes in 1996~\cite{Luby97,Robinson02}. Tornado codes have linear decoding cost of $n\log (\frac{1}{\epsilon})$, and transmission rate of 1.06. Software-based implementations of Tornado codes were shown to be about 100 times faster on small input packet size and about 10,000 times faster on larger input packet size than other software based implementation of RS codes~\cite{Byers98}.

Tornado codes are constructed using concatenation at different level. At the first layer, the input packets are XOR coded into coded packets of smaller length, these coded packets are then in turn XOR coded into coded packets at the second layer and so on. The last layer is coded using conventional code such as RS codes. Tornado codes are patented~\cite{Tornado0, Tornado1}.

Unfortunately Tornado codes are \textit{fixed-rate} codes, i.e. once the encoder has chosen $k$ and $n$ based on initial channel erasure rate estimation, the encoder can generate only $n$ codewords. Therefore if the average channel erasure rate is less than what the encoder initially estimated, then this would lead to redundant codewords at the decoder, whereas if the average erasure rate is higher than what was initially estimated, then that would lead the decoder unable to decode due to insufficient codewords. However for most internet applications and the wireless channel, the erasure rate has been shown to be stochastic in nature~\cite{Aguayo04}. This motivates the
design of rateless codes with linear decoding cost.

\subsection{Luby-Transform and Raptor Codes}

Armed with Tornado codes, Luby along with Goldin founded the Digital Fountain company to develop efficient erasure codes in 1998, drawing capital investment from Adobe, Cisco Systems and Sony Corporation~\cite{Robinson02}. The company and the codes the company developed are known as ``Fountain,'' based on what the codes achieve, generate virtually unlimited supply of codewords, analogous to a Fountain producing limitless drops of water. Just as an arbitrary collection of water drops will fill a glass of water and quench thirst, irrespective of which water drops had been collected. Collection of any $n$ Fountain codewords will be sufficient for the decoder to decode the input packets.

The main idea behind Luby-Transform (LT) codes and Raptor codes is to design the \textit{degree distribution} (such as the robust Soliton distribution, or one of its variant) of a coded packet. The \textit{degree} of a coded packet, indicates the number of input packets used to generate the coded packet. LT codes and Raptor codes are irregular codes. LT coding is done in two steps, first the encoder randomly selects the degree, whose expected probability is dictated by the degree distribution. In the next step the encoder, randomly selects input packets, the number of input packets selected is given by the degree selected in the first step, and perform XOR addition of those input packets.

Decoding is performed using back-substitution. The decoder looks for those coded packet with one unknown input packets, and decode the unknown packet. It then substitutes this decoded packets in all the other coded packets which had been generated using this decoded packet as one of its constituent encoding packet. The decoder continues to repeat this process of decoding and substitution until it has not decoded all the $k$ input packets. If it is unable to decode $k$ input packets, then it requests for additional coded packets to be transmitted by the server.

Raptor (Rapid Tornado) codes is a special class of LT codes. The design of Raptor codes is motivated by the fact that due to the random nature of selecting the input packets, there is always a non-zero probability that some of the input packets may never be selected for coding in LT codes. To address this problem, in Raptor code the input symbols are first \textit{precoded}, and then LT coding procedure takes place. Precoding can be thought of concatenation of input packets $(c_1, c_2, \ldots, c_k)$ and redundant packets $(y_1, y_2, \ldots, y_j)$, given as $(c_1, \ldots, c_k,y_1,\ldots, y_j)$. The redundant packets can be generated by randomly coding the input packets using XOR addition. After the precoding the output packets are generated using LT coding procedure, whose input packets are given by the concatenation $(c_1, \ldots, c_k,y_1,\ldots, y_j)$.

Decoding Raptor codes is done using inactivation decoding. While a detailed illustrating example of decoding Raptor code is given in~\cite[250-255]{Raptor}, we here give the main idea behind inactivation using simpler example. Consider the coding coefficient matrix for the Raptor codes in Figure~\ref{fig:inactivation}. Instead of performing Gaussian elimination on the complete matrix, inactivation algorithm first attempts to make the matrix sparse by back-substituting 1-sparse rows and the resulting 1-sparse rows. When no more 1-sparse coded packets are present, the algorithm decodes the resulting coded packets represented by a smaller but denser submatrix. Since Gaussian elimination runs on a much smaller submatrix, the resulting overall decoding complexity can be reduced.

\begin{figure}
\begin{center}
\includegraphics[width = 0.5\textwidth]{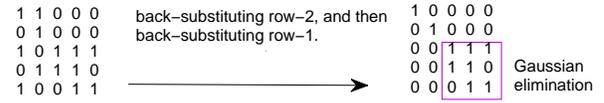}
\end{center}
\caption{An example to illustrate inactivation decoding.} \label{fig:inactivation}
\end{figure}

\subsection{LT and Raptor Codes Performance}

Luby-Transform (LT) codes~\cite{Luby02} and Raptor codes~\cite{Amin06}, along with Tornado codes, are classes of $\mathbb{F}_2$ linear codes broadly known as
\textit{Fountain codes}. Both LT and Raptor codes are rateless codes, with asymptotic decoding complexity of $\mathcal {O}(k\log(\frac{k}{\delta}))$ and $\mathcal {O}(k\log(\frac{1}{\epsilon}))$ respectively to deliver asymptotic optimality, where $1-\delta$ is the probability that the LT decoder can recover the input packets from $n$ codewords.

However for finite input packet length, particularly small values of $k$, LT and Raptor codes suffer from degrading transmission rate performance. Even for
large input packet length of $k\approx 10,000$, LT codes have an overhead of 5\%~\cite{Mackay05}, while for packet length of $k=65,536$, Raptor codes have an overhead of 3.8\%~\cite{Amin06}. Asymptotic rate optimality for LT and Raptor codes is therefore achieved only when the order of $k$ is in 100,000's of packets.

In addition, both LT codes and Raptor codes are not systematic codes, therefore limiting its application for scenarios such as the index coding problem~\cite{Chaudhry08} and cooperative data exchange problem~\cite{Tajbakhsh11}, where the decoder should be able to decode the input packets using a combination of subset of the input packets and coded packets. Similarly for Peer-to-Peer Content Distribution Network (P2P-CDN), a user may want to ``preview'' the content of the file before dedicating several hours to download the complete file, to verify that an incorrect file has not been uploaded by a dishonest user. The use of systematic codes is also a requirement in data storage, where the storage device should be able to recover an erased input packet, by using input packets and coded packet (also known as parity packets).

Hyytia \textit{et al.} studied optimal degree distribution for LT over small input packet size of $k\leq 30$~\cite{Hyttia07}. The results concluded that even with optimized degree distribution, LT codes have a redundancy of approximately 40\% for $k\leq 20$. These optimized degree distribution were calculated using a recursive equation, with exponential running time, making it computationally impractical to design optimal degree distribution for $k>20$.

Patent description of LT codes are given in~\cite{LT0,LT1}, while those of Raptor codes are given in~\cite{Raptor0, Raptor1}.

\subsection{Standardized Raptor Codes}
Digital Fountain was later acquired by Qualcomm in 2009. During this time period, Digital Fountain went on to standardize the Raptor codes. These standardized and patented versions of Raptor codes are known as the Raptor 10 (R10) and Raptor Q (RQ) codes~\cite[Chapter 3]{Raptor}~\cite{RQcodes}, and are systematic rateless codes designed to provide near-optimal transmission rates for finite length input packets.

The R10 codes generates some coded packets similar to RL codes over $\mathbb{F}_2$, which are relatively denser. The remaining coded packets are sparsely coded based on the specified precoding and LT coding procedure. The performance of R10 codes can be approximated and upper bounded by the performance of RL codes over $\mathbb{F}_2$. Analytical results of RL codes over $\mathbb{F}_2$ are given in~\cite{Cruces11}, while the performances of sparse RL codes over $\mathbb{F}_2$ are reported in~\cite{Studholme10}. The RQ codes differs from R10, in that instead of generating some coded packets using RL codes over $\mathbb{F}_2$, it uses a larger field size of $\mathbb{F}_{256}$. Similarly the performance of R10 codes are bounded by the performance of RL codes over $\mathbb{F}_{256}$. Performance evaluation of R10 and RQ codes can be found in~\cite[Chapter 3]{Raptor}.

Such improvements in performance come at the tradeoff cost of using the inactivation decoding algorithm which uses combination of Gaussian elimination and belief-propagation decoding algorithm to decode the codewords. The decoding complexity of R10 and RQ codes is given as $\mathcal {O}(k^{1.5})$~\cite[pp. 254-255]{Raptor}~\cite[pp. 8]{RQcodes}. Practical implementation of R10 codes has shown that Gaussian elimination can account for up to $92\%$ of the total decoding time even for a modestly large value of $k=1024$, where each packet is 4 bytes long~\cite{Todor11}.

Even with the decoding complexity of $\mathcal {O}(k^{1.5})$, R10 codes can only support up to 8,192 input packets, while RQ codes can support up to 56,403 input packets. This limitation comes from the design of degree distribution of the codewords for R10 and RQ codes. Once the number of input packets exceeds these limits, the decoding failure probability gradually increases.

\section{Triangular Codes}\label{sec:triangular}
\begin{table}
\centering \caption{\label{table:example1} Packet reception status example.}
\begin{tabular}{|c|c|c|c|}
\hline clients/packets & $c_1$ & $c_2$ & $c_1\oplus c_2$ \\
\hline $r_1$ & received & $\times$ & $\times$ \\
\hline $r_2$ & $\times$ & received & $\times$ \\
\hline $r_3$ & $\times$ & $\times$ & received \\
\hline 
\end{tabular}
\end{table}

A new class of nonlinear erasure codes, which can be in practice be encoded and decoded using XOR addition operation despite being nonlinear, called Triangular codes, was proposed recently to further improve the performance while maintaining low decoding complexity. This class of codes uses a mix of finite field and real field operations for encoding, and the encoded packets can be decoding using back-substitution method involving only XOR addition. The tradeoff cost of Triangular codes is the additional redundant bits needed to pad to each packet. However, it has been demonstrated that a small number of additional bits is sufficient to lead the codes to a near-optimal performance of transmission rate while maintaining low decoding complexity for finite length input packets. A proof-of-concept for Triangular codes was first presented in~\cite{Qureshi12}. A patent application for these codes has also been submitted~\cite{Triangular0}. Research and development for Triangular codes is currently an ongoing work.

\subsection{Motivating Example}
We illustrate the motivation of Triangular codes using an example. Consider for example, a server multicasting packet $c_1$ and $c_2$ to three clients. After the transmission on independent BEC, $c_1$ is received only by $r_1$, and $c_2$ is received only by $r_2$. Now the only linearly independent packet the server can transmit under the constraint of $\mathbb{F}_2$ linear coding is $c_1\oplus c_2$. Assuming that the transmitted $c_1\oplus c_2$ packet is received only by $r_3$ as shown in Table~\ref{table:example1}, then with such packet reception status, there is no way the server can generate a $\mathbb{F}_2$ linear code which is linearly independent for all the three clients.

To go around the problem depicted in Table~\ref{table:example1}, in our design, redundant `0' bits are selectively added at the head and tail of the input packets before performing XOR addition on packets. Assume that all packets $c_i$, $i\leq k$, are $B$-tuple, denoted as $c_i\triangleq (b_{i,1},b_{i,2},...,b_{i,B})$, $b_{i,j}\in\{0,1\}$. Bit $b_{i,j}$ is read as the $j^{th}$ bit of the $i^{th}$ packet. Hence the bit pattern of $c_1$ with one redundant `0' bit added at the head of the packet is given by the tuple $(0,b_{1,1},...,b_{1,B}),$ and that of $c_2$ as $(b_{2,1},...,b_{2,B},0)$. The packet header will include information about the number of redundant `0' bits added at the tail of each packet used to generate the coded packet.

The modified packets $c_1$ and $c_2$ can now represented as $c_{1,1}$ and $c_{2,0}$, where $c_{i,\upsilon}$ is read as the $i^{th}$ packet with $\upsilon$-0 bits added at its head. The transmitted packet $c_{1,1}\oplus c_{2,0}$ will be linearly independent for all the three clients. An illustration of the encoding $c_1\oplus c_2$ and $c_{1,1}\oplus c_{2,0}$, corresponding to $c_1\oplus 2\times c_2$ is shown in Figure~\ref{fig:code}. 

If $r_3$ receives $c_{1,1}\oplus c_{2,0}$, then using packet $c_1\oplus c_2$, it now has information about bit $b_{2,1}$, and $b_{1,1}\oplus b_{1,2}$. Using bit $b_{2,1}$ it can decode bit $b_{1,1}$ from $b_{1,1}\oplus b_{1,2}$. Bit $b_{1,1}$ is then substituted in $b_{1,1}\oplus b_{2,2}$, from the encoded packet $c_{1,1}\oplus c_{2,0}$, to obtain bit $b_{2,2}$, which can then be substituted in $b_{1,2}\oplus b_{2,2}$ to decode $b_{1,2}$. Therefore using this bit-by-bit simple back substitution method, $R_3$ can decode all the bits of packet $c_1$ and $c_2$. Clients $r_1$ and $r_2$ can also similarly decode the unknown packet from $c_{1,1}\oplus c_{2,0}$.

\begin{figure}
\begin{center}
\includegraphics[width = 0.5\textwidth]{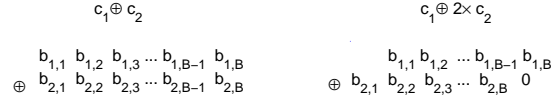}
\end{center}
\caption{An example to illustrate encoding process for $c_1\oplus c_2$ and $c_{1,1}\oplus c_{2,0}$. The presence or absence of the added bit `0' at the head of $c_{1,1}$ does not effect the actual encoded bitstream, as long as the right most bits of the two packets are aligned.} \label{fig:code}
\end{figure}

\subsection{Encoding and Decoding}
The implementation of adding redundant `0' bits can be described by the real field multiplication of the bitstream tuple with $2^\ell$, $\ell\in \mathbb{N}_0$, where $\ell$ denotes the number of redundant `0' bits added at the tail of the packet. After the multiplication, with the right most bits (tail) of the concerned input packets aligned, the packets are XOR added to generate a coded packet. To equalize the length of all packets for this finite field addition operation, redundant `0' bits are padded at the head of various packets, such that the total length of all packets used for encoding is equal. The encoded packet can be given as $c_1\oplus (2\times c_2)$, i.e. with coding coefficients given as $[1, 2]$. The encoding procedure for generating a coded packet $\kappa$ from the set of packets $\mathcal{P}$ using Triangular coding scheme can be given as,

\[
\kappa = \bigoplus_{c_i\in \mathcal{P}} 2^{\ell}\times c_i.
\]

The use of real field multiplication avoids the complicated finite field multiplication in the encoding. Real field multiplication can be easily accomplished by logical shift operation. While in the decoding, back-substitution can be used immediately on the set of encoded packets. The idea and design challenge of Triangular codes is to add redundant `0' bits to each packet such that there exists a row interchange permutation whereby these redundant `0' bits form a triangular pattern. If in addition to forming a triangular pattern, the coding vector also forms a full rank matrix, then the packets can be decoded using back-substitution only. Since in the coding coefficients, the unknown bits already form a Triangular matrix, the triangularization step of Gaussian elimination is no longer needed. The decoding complexity is therefore bounded as $\mathcal{O}(k^2)$ for each bit location using back-substitution. The complexity may be further reduced to $\mathcal{O}(k\omega)$, $\omega\leq k$, with a design of suitable sparse coding coefficients. As overhead of additional redundant bits are needed to pad to packets, the design challenge of Triangular codes is strike a balance between performance in retransmission rate and overhead for specific applications.

\section{Future Research Directions}\label{sec:future}
In this section we discuss few open problems where research on erasure coding is ongoing, namely, decoding delay of erasure codes, pollution attack, and the index coding problem.

While the use of erasure code improves the bandwidth performance of a broadcast network, it has a disadvantage of incurring a decoding delay. For example, for a client which has packet $c_1$ and wants packets $c_2$ and $c_3$, coded packets $c_2\oplus c_3$ and $c_1\oplus c_2$ are both linearly independent with respect to $c_1$, however only the latter coded packet can be instantly decoded by the client. When considering RS and RL codes, a client need to have $k$ linearly independent packets before it can start the decoding process. Designing an erasure coding scheme such that the time duration the client need to wait before it can decode the coded packet is an ongoing open research topic. We refer interested readers to~\cite{Keller08} and references therein.

Another problem related to security aspects of erasure codes is that of pollution attack. If the client admits even a single malicious coded packet from a malicious user, then during the decoding process, all the decoded packets will be corrupted, and hence result in turning the correctly used coded packet as being useless, if no appropriate security mechanism is in place. Devising algorithmic approach to prevent such pollution attack, and its corresponding computational complexity is also an ongoing research topic. We refer interested readers to~\cite{Yu09} and references therein.

An overview of the index coding problem had been presented earlier in Section~\ref{subsect:related}.

\section{Conclusion}\label{sec:conclusion}
Traditional approaches to deal with system erasure are to use retransmission and replication techniques, which limits the reliability of the system, and adversely effects the throughput performance of the system. In this paper, we presented motivation for the use of erasure coding to improve the reliability and performance of data transmission and data storage system. However the principle disadvantage of using traditional erasure coding such as the Reed-Solomon coding is that such codes have high decoding complexity limiting its implementation, especially on battery and processor constrained devices such as smartphones.

In the last two decades, motivated by the exponential increase in the data traffic over the Internet, a series of Fountain codes - Tornado codes, LT codes, and Raptor codes - have been proposed, and patented, to address the decoding complexity of RS codes. Unfortunately Fountain codes can achieve linear decoding complexity only when the input packet length is asymptotically large. For smaller input packet length, Fountain codes suffer from degrading transmission rate.

To address this tradeoff, recently the Triangular codes had been proposed. Triangular codes are non-linear codes, where the encoder add redundant bits at the head and tail of a packet before performing $\mathbb{F}_2$ addition of the packet. The main idea behind Triangular codes is to code the packets such that apart from being linearly independent with high probability, such coded packets can be decoded using the back-substitution decoding algorithm.

With the increasing use of mobile devices for Internet access, erasure codes will certainly continue to play an important roles wireless data transmissions. More research works are warranted to further reduce the transmission overhead while maintaining near-optimal coding performance.

\bibliographystyle{IEEEtran}
\bibliography{IEEEabrv,mainJ}

\end{document}